\documentclass[journal]{IEEEtran}

\usepackage{setspace}

\usepackage{graphics,
           psfrag,
           epsfig,
           amsthm,
           cite,
           amssymb,
           url,
           dsfont,
           subfigure,
           algorithm,
           algorithmic,
           balance,
           enumerate,
           color,
           setspace
}
\usepackage{amsmath}

\usepackage{epstopdf}

\newtheorem{definition}{Definition}

\newtheorem{lemma}{Lemma}

\newtheorem{theorem}{Theorem}

\newtheorem{conjecture}{Conjecture}

\newtheorem{remark}{Remark}

\newcommand{\sref}[1]{Section~\ref{#1}}
\newcommand{\appref}[1]{Appendix~\ref{#1}}
\newcommand{\fref}[1]{Figure~\ref{#1}}

\newcommand{\cref}[1]{Constraint~\ref{#1}}
\newcommand{\thref}[1]{Theorem~\ref{#1}}

\newcommand{\lref}[1]{Lemma~\ref{#1}}

\newcommand{\algref}[1]{Algorithm~\ref{#1}}

\newcommand{\ignore}[1]{}

\IEEEoverridecommandlockouts
\begin{document}

\title{\vspace{-.5cm}Robust Node Estimation and Topology Discovery Algorithm in Large-Scale Wireless Sensor Networks}
\author{
\authorblockN{Ahmed Douik$^{\dagger}$, Salah A. Aly$^\ast$, Tareq Y. Al-Naffouri$^{\prime\star}$, and Mohamed-Slim Alouini$^\prime$\\}%
\authorblockA{$^\dagger$California Institute of Technology (Caltech), California, United States of America \\
$^\ast$ Fayoum University (FU), Zewail City of Science and Technology, Egypt\\
$^\star$King Fahd University of Petroleum and Minerals (KFUPM), Kingdom of Saudi Arabia \\
$^\prime$King Abdullah University of Science and Technology (KAUST), Kingdom of Saudi Arabia \\
Email: $^\dagger$ahmed.douik@caltech.edu, $^\ast$salah@zewailcity.edu.eg, \\
$^\prime$\{tareq.alnaffouri,slim.alouini\}@kaust.edu.sa }
\vspace{-.8cm} }

\maketitle

\IEEEoverridecommandlockouts

\begin{abstract}
This paper introduces a novel algorithm for cardinality, i.e., the number of nodes, estimation in large scale anonymous graphs using statistical inference methods. Applications of this work include estimating the number of sensor devices, online social users, active protein cells, etc. In anonymous graphs, each node possesses little or non-existing information on the network topology. In particular, this paper assumes that each node only knows its unique identifier. The aim is to estimate the cardinality of the graph and the neighbours of each node by querying a small portion of them. While the former allows the design of more efficient coding schemes for the network, the second provides a reliable way for routing packets. As a reference for comparison, this work considers the Best Linear Unbiased Estimators (BLUE). For dense graphs and specific running times, the proposed algorithm produces a cardinality estimate proportional to the BLUE. Furthermore, for an arbitrary number of iterations, the estimate converges to the BLUE as the number of queried nodes tends to the total number of nodes in the network. Simulation results confirm the theoretical results by revealing that, for a moderate running time, asking a small group of nodes is sufficient to perform an estimation of $95\%$ of the whole network.
\end{abstract}


\section{Introduction} \label{sec:int}

Wireless Sensor Networks (WSNs) have been a great success in the past decades. Generally, a WSN refers to a set of small electronic devices (sensors) capable of monitoring and measuring certain phenomena, e.g., temperature, pressure, flood, fires, etc., usually in hazardous and non-reachable environments. A WSN is typically composed of hundreds to millions of nodes capable of intersecting and communicating with each other. Due to their small size, these devices have limited resources such as memory, computation power, battery lifetime and bandwidth.

This paper is interested in estimating the cardinality, i.e., the size, of a network. In other words, the aim is to determine the number of nodes distributed randomly and uniformly in a given field. There are several benefits for cardinality estimation in graphs such as energy efficiency \cite{Varagnolo2014}, mobile communication and coding schemes design \cite{Cattani2014}, and distributed storage \cite{aly08e,aly09f}. Furthermore, the paper proposes that each node discovers its neighbours which help network designers enhancing coverage and connectivity \cite{Shames2012}.

Applications of the network size estimation are not limited to WSN. With the shift in the design from the centralized architectures to decentralized ones, the problem becomes increasingly in demand due to its applications in social networks and artificial intelligence \cite{Katzir2014,Wang2013,shafaat2008practical}. Even though decentralized systems are more scalable and robust to failure, their use makes the estimation of the parameters of the whole network challenging.

For large-scale networks, where network size can reach up to couple of million nodes, it is computationally expensive to brute-search all nodes to infer information about the entire network. Moreover, it is infeasible that each node communicates with the data collector (DC). In order to determine the size of such massive systems, two trends can be distinguished in the literature, namely the node counting and the size estimation.

The problem of node counting in an undirected graph is introduced in \cite{Szymanski1998}. Generally, the aim is to visit all nodes of a graph while avoiding at maximum revisiting nodes. The problem has numerous applications in artificial intelligence and control theory. However, the authors in \cite{Szymanski1998} prove it to be NP-complete with a time complexity given by $\Omega{(n^{\sqrt{n}})}$ where $n$ is the cardinality of the graph. Therefore, such approach it unsuited for large-scale networks.

The use of statistical inference for cardinality estimation of a given network first appears in the literature with the German tank problem \cite{Gum2005105}, in which the aim is to estimate the total number of tanks given the serial number of the captured ones. The fundamental idea of cardinality estimation is to sample a subset of the entire population and available information. In other words, querying only a small portion of nodes to infer information about the whole network status.

\subsection{Related Work}

Due to the complexity of nodes' counting, numerous research works focus on estimating the cardinality and mean edges degree in a graph. The authors in \cite{aly08e,aly10c} state a method for determining the total number of nodes in a graph using data flooding and random walks. They present an algorithm for node estimation in large-scale wireless sensor networks using random walks that travel through the network based on a predefined probability distribution.

Ribeiro et al. \cite{Ribeiro2012b,Ribeiro2012} propose a scheme for estimating the network parameters in directed graphs using random walks. Their algorithm precisely predicts the out-degree distributions of a variety of real-world graphs. Similarly, Dalal et al. \cite{dalal05} study the problem in a context of robust visual object recognition. The authors in \cite{Katzir2014} present two algorithms for estimating the total number of online users in social networks by using sampling from graphs assumed to have a stationary distribution.

The authors in \cite{Cooper2012,Dasgupta2014} propose a model for estimating the network parameters using random walks in graphs. In particular, by sampling from the graph, they propose a method to determine the average edge degree rather than the individual node degree. The problem of estimating the mean degree of a graph is first suggested by Feige et al. \cite{feige2006}. The authors in \cite{Zhang2014} present a sampling method for node degree estimations in a sampled network and the authors in \cite{Wang2013} show a way to obtain content properties by testing a small set of vertices in the graph.

While the authors in \cite{Varagnolo2014} propose a model for distributed cardinality estimation in anonymous networks using statistical inference methods, the authors in \cite{Cohen1997} present a scheme for estimating the number of reachable neighbours for a given node and a size of the transitive closure. They present an $O(n)$ time complexity algorithm based on Monte Carlo that estimates, with a small error, the sizes of all reachability sets and the transitive closure of a graph.

\subsection{Contribution}

The difficulty of the network size estimation heavily depends on the assumptions and features of the system. This paper considers the anonymous networks framework \cite{shafaat2008practical}, where nodes only know their unique identifier (ID). The authors in \cite{481599} show that with a centralized strategy, the node estimation can be obtained in finite time with probability one. For non-unique IDs, the authors in \cite{cidon1992message,5978198} demonstrate that the estimation cannot be performed with probability one in limited time or with a bounded computational complexity. The problem is, then, to discover estimators that trade-off small error likelihood and moderate computational complexity.

This paper proposes a hybrid scheme that not only performs node counting that can be run for an arbitrary time rounds but further uses the output to carry out the network size estimation. Such estimate benefits network designers to design the coding schemes appropriately. The proposed system combines the advantages of both the node counting algorithms and the node estimation ones. Given that the time and computation complexity of node counting algorithms is high, their use in large-scale networks is prohibitive. On the other hand, network size estimation algorithms have, in general, high variance. Depending on the initialization parameters, the estimate of the proposed scheme balance these two effects and can be made arbitrary as close to the network cardinality as wanted. Furthermore, the algorithm suggests, at the same time, to discover the neighbours of each node. Such knowledge is crucial for data routing that can be combined with the code design resulting in efficient resource utilization. Due to space limitation, this work considers nodes with unique IDs. However, this assumption can be removed in a future work by exploiting the inverse birthday paradox.

The rest of this paper is organized as follows: In \sref{sec:net}, the system model and the problem formulation are presented. \sref{sec:pro} illustrates the proposed cardinality estimation algorithm whose performance analysis are characterized in \sref{sec:per}. Simulation results are shown and discussed in \sref{sec:sim}. Finally, \sref{sec:con} concludes the paper.

\section{Network Model and Problem Formulation} \label{sec:net}

\subsection{Network Model}

Consider a wireless sensor network $\mathcal{N}$ with $n$ sensor nodes that are randomly and uniformly distributed in a region $A=[0,L] \times [0,W]$ for some $L,W > 0$. The network $\mathcal{N}$ can be considered as an abstract graph $\mathcal{G}=(\mathcal{V},\mathcal{E})$ with a set of nodes $\mathcal{V}$ and a set of edges $\mathcal{E}$, where $n=|\mathcal{V}|$. The set $\mathcal{V}= \{s_1,\ \cdots,\ s_{n}\}$ represents the sensors that measure information about a specific field, and $\mathcal{E}$ represents the set of links between the sensors.

Two arbitrary sensors $s_i$ and $s_j$ for $1 \leq i \neq j \leq n$ are connected if they are in the transmission range each other. Assuming that the transmission range is circular, let $R$ be its radius\footnote{The algorithm is independent of the considered transmission range. However, the performance analysis provided in the rest of the paper assumes circular transmission range with the same radius for all nodes}. Therefore, $s_i$ and $s_j$ are connected if and only if $d(s_i,s_j) \leq R$, where $d(.,.)$ is the distance operator.

The paper assumes that neither the number of these $n$ nodes, i.e., the network size, nor their connections, i.e., the network topology, are known. However, a bound on the network size $N_{\max} \geq n$ is known by the data fusion center. This scenario can be seen as a network after a long running time or a disaster. Initially, the network is composed of $N_{\max}$ nodes each having a unique ID. After a long running time or a disaster, some of the nodes may disappear from the graph leaving a graph with $n \leq N_{\max}$ nodes with unique ID. Let $ID_i$ be the ID of sensor $s_i$.

\subsection{Network Protocol}

In the considered network model, each node knows only its unique identifier. Communication between nodes is performed by broadcasting the information to transmit. Note that nodes needs not to transmit additional bits indicating its ID with the information packet. Moreover, no acknowledgement is expected from sensors that successfully receive a packet. Transmissions are subject to erasure at the sensors with a probability $q_{s_i}$ for sensor $s_i$. In other words, for a sensor $s_i$ broadcasting data, sensors $s_j \in S_i$ successfully receives the data with probability $1-q_{s_j}$ where $S_i$ the set of neighbours of a node $s_i$ defined as follows.

\begin{definition}
Denote by $S_i$ the set of neighbours of a node $s_i, 1 \leq i \leq n$. In other words, $S_i = \{s_j \in \mathcal{V}$ such that $d(s_i,s_j) \leq R\}$.
\end{definition}

This paper consider static nodes in the network. Therefore, due to the motion-less of nodes, their relative position in the network remains identical which results in an unchanged set of neighbours for all nodes.

\subsection{Problem Formulation}

Given the aforementioned network model and protocol, this paper's objectives is to:
\begin{enumerate}
\item Estimate the number of nodes $n$ by asking $K$ randomly nodes in the network. Let $\mathcal{K}$ be the set of nodes in the network that can be queried by the data collector. This set of nodes is randomly picked from the set of alive and dead nodes with $|\mathcal{K}|=K \ll |\mathcal{V}|=n \leq N_{\max}$. 
\item Discover locally for an arbitrary node $s_i$ its set of neighbours $S_i$.
\end{enumerate}

Without a loss of generality, the DC is assumed to know the IDs of the nodes in the initial network comprising $N_{\max}$ sensors. The selection of the queried nodes is performed by sampling uniformly without replacement from this set of IDs. Such methods results in $K$ nodes randomly picked from the set of alive and dead ones. Throughout the paper, the notation $\mathcal{U}(0,1)$ refers to the uniform distribution over $(0,1)$.

\section{Proposed Node Estimation Algorithm} \label{sec:pro}

This section introduces the hybrid node counting and estimation algorithm. The algorithm estimates the total number of nodes in a network. The algorithm runs in three distinct phases: the initialization, the knowledge distribution, and the query phases. In the initialization phase, the initial packets of the nodes and their transmit probability are set. In the knowledge distribution phase, the information about the networks is disseminated among the surviving nodes from neighbour to neighbours. Finally, in the query stage, the DC collects the information about the network by asking some nodes and inferring the size of the whole system.

\subsection{Initialization Phase}

\begin{algorithm}[t]
\begin{algorithmic}
\REQUIRE $\mathcal{G}=(\mathcal{V},\mathcal{E})$, with $\mathcal{V}=\{s_1,\ldots,s_{n}\}$ and $f_{s_i}^{\text{initial}},\ \forall \ s_i$.
\STATE Initialize $\mathcal{T}(0) = \varnothing$.
\FORALL {$s_i \in \mathcal{V}$}
\STATE Initialize $P_{s_i}=\{ID_i\}$
\STATE Initialize $f_{s_i}=f_{s_i}^{\text{initial}}$
\ENDFOR
\end{algorithmic}
\caption{Initialization Phase.}
\label{alg:initial}
\end{algorithm}

In the initial step, each node in the network generates a packet containing its ID. As the packet size limitation is crucial, this paper consider reducing it. For a network with initial $N_{\max}$ nodes, the distinct IDs can be encoded using $\lceil \log_2 (N_{\max}) \rceil$, where $\lceil.\rceil$ is the ceiling function. Therefore, the maximum size a packet can reach at any node in the network is $n \lceil \log_2 (N_{\max}) \rceil$ as only $n$ nodes are alive. Such packet size is convenient for practical scenarios as it scales logarithmically with $N_{\max}$ and linearly with $n$.

Each node $s_i$ also initializes its initial transmit probability $f_{s_i}^{\text{initial}}$, where $f_{s_i}^{\text{initial}}$ is the probability that the node broadcasts the packet it already holds to its neighbours. Whereas a small value of the initial probability means that there is small amount of communication between nodes in the network, a value $f_{s_i}^{\text{initial}} \approx 1$ means that all nodes broadcast their packets at each iteration. Let $\mathcal{T}(t)$ be the set of nodes that transmitted a packet at time instant $t$ with $\mathcal{T}(0) = \varnothing$. \algref{alg:initial} summarizes the steps of the initialisation phase.

\begin{remark}
The proposed algorithm can be easily extended to perform topology discovery, i.e., the estimation of both $\mathcal{V}$ and $\mathcal{E}$, by modifying the initial packets of each node. Each node $s_i$ generates a packet containing both its ID and its $(X,Y)$ coordinates. Assuming that coordinates are encoded using $V$ bits, e.g., $V=32$ bits to encode a real number, the maximum size a packet can reach is $2 V n \lceil \log_2 (N_{\max}) \rceil$. Therefore, the size of the topology discovery packet scales in the same manner as the one of the cardinality estimation. Due to space limitations, the performance analysis of the topology discovery scheme is omitted in this paper as it follows similar steps to the ones exposed herein.
\end{remark}

\subsection{Knowledge Distribution Phase}

\begin{algorithm}[t]
\begin{algorithmic}
\REQUIRE $\mathcal{G}=(\mathcal{V},\mathcal{E})$, with $\mathcal{V}=\{s_1,\ldots,s_{n}\}$.
\FOR{$t=1,\ 2,\ \cdots$}
\STATE Set $\mathcal{T}(t) = \varnothing$.
\FORALL {$s_i \in \mathcal{V}$}
\FORALL {$s_j \in \mathcal{T}(t-1)$}
\IF{$P_{s_j}$ heard}
\STATE Set $S_i$ = $S_i \cup s_j$.
\IF{$P_{s_j} \nsubseteq P_{s_i}$}
\STATE Set $P_{s_i} = \left( P_{s_i} \cup P_{s_j} \right)  \setminus ID_i$.
\STATE Set $P_{s_i} = \{P_{s_i},ID_i\}$.
\STATE Set $f_{s_i} = \cfrac{1}{2}(f_{s_i}+1)$.
\ENDIF
\ENDIF
\ENDFOR
\STATE Sample $u_{s_i}$ from $\mathcal{U}(0,1)$.
\IF{$u_{s_i} < f_{s_i}$}
\STATE $s_i$ broadcasts $P_{s_i}$.
\STATE Set $f_{s_i}=f_{s_i}^{\text{initial}}$.
\STATE Set $\mathcal{T}(t) = \mathcal{T}(t) \cup s_i$
\ENDIF
\ENDFOR
\ENDFOR
\end{algorithmic}
\caption{Knowledge dissemination Phase.}
\label{alg:knowledge}
\end{algorithm}

In this phase, the knowledge is distributed among the alive nodes in the network from neighbours to neighbours. At each running time of the algorithm, if a node $s_i$ receives a packet from a node $s_j$ whose ID can be determined by examining the last ID in the received packet, it adds such node to is set of neighbours $S_i$. Depending on the content of the received packet, two scenarios can be distinguished:
\begin{itemize}
\item The packet does not contain a new information for $s_i$, i.e., $P_{s_j} \subseteq P_{s_i}$): The packet is discarded and the buffer is not updated.
\item The packet brings a new information to the node, i.e., $P_{s_j} \nsubseteq P_{s_i}$: The node update its buffer and increases its transmit probability. The more innovative packets a node receives, the more its transmit probability increases. This is motivated by the fact that the more new information a node receives, the better candidate it is to transmit. To be able to estimate locally the neighbours of the nodes, each node first remove its ID from the packet it possesses and then append it to the end of the packet.
\end{itemize}

Afterward, each node $s_i$ samples from a probability distribution $\mathcal{U}(0,1)$ and decides, according to $f_{s_i}$, either to broadcast $P_{s_i}$ or not. After broadcasting data, the node resets its transmit probability to the initial value. This is motivated by the fact that after transmission, if all neighbours received the packet, then node $s_i$ does not bring new information anymore unless it receives new packets. \algref{alg:knowledge} summarizes the steps of the knowledge distribution phase.

\subsection{Query Phase}

\begin{algorithm}[t]
\begin{algorithmic}
\REQUIRE $\mathcal{G}=(\mathcal{V},\mathcal{E})$ and $\mathcal{K}$ with $|\mathcal{K}| = k \ll n$.
\STATE Initialize $\tilde{P} = \varnothing$.
\FORALL {$s_i \in \mathcal{K}$}
\STATE $\tilde{P} = \tilde{P} \cup P_{s_i}$.
\ENDFOR
\STATE Set $\tilde{n} = |\tilde{P}|$
\end{algorithmic}
\caption{Data queries and network size estimation.}
\label{alg:collection}
\end{algorithm}

In this phase, a DC queries some nodes from the set of nodes to retrieve information about the current status of the network $\mathcal{N}$ and infer its size. If the queried node $s_i$ is alive then it transmits its packet $P_{s_i}$. Otherwise, there is no transmission, and the packet of that node is the empty set.

After querying the nodes, their packets are processed using the union operator and by counting the number of IDs. In other words, the quantity $\tilde{Z}$, the counting estimation, can be obtained by $\tilde{Z} = |\tilde{P}|$ where $|.|_1$ is the cardinality operator. \algref{alg:collection} summarizes the steps of the data collection and network size estimation phase. The next section relates the counting estimation to the Best Linear Unbiased Estimators (BLUE) of network size $\tilde{n}$.

\section{Performance Analysis} \label{sec:per}

Let $X_{ij}(t)$ be a Bernoulli random variable denoting if node $s_i$ knows that node $s_j$ is alive. Let $X_i(t) = (X_{i1}(t),\ \cdots, \ X_{iN_{\max}}(t))$ be the vector containing the knowledge of node $s_i$. From \algref{alg:knowledge}, $P_{s_i}$ is the realisation of the random variable $X_i(t)$ at each time slot $t$.

Let $Z(t)=(Z_1(t), ,\ \cdots, \ Z_{N_{\max}}(t))$ be a random variable where $Z_i(t), \ 1 \leq i \leq N_{\max}$ is a Bernoulli random variable denoting if the central unit knows that node $s_i$ is alive when the data collection is performed at time slot $t$. Let $\tilde{Z}(t) = \sum\limits_{i=1}^{N_{\max}} Z_i(t)$. From \algref{alg:collection}, $\tilde{Z}$ is the realisation of $Z(t)$ at query time $t$. Given the data collection equation, the random variable $Z_i(t), \ 1 \leq i \leq N_{\max}$ can be written as follows:
\begin{align}
Z_i(t) = \max_{s_j \in \mathcal{K}} X_{ji}(t).
\end{align}

Define $\mathcal{A}$ as the set of node that are alive and $\mathcal{D}=\mathcal{N}\setminus \mathcal{A}$ the set of nodes that are dead where $\mathcal{N}$ is the set of all nodes in the network. It can be easily seen that $|\mathcal{N}|=N_{\max}$ and $|\mathcal{A}|=n$.
\begin{definition}
Let $\mathcal{B}_t(s_i), t \geq 1$ be the $t$-degree neighbours function defined as:
\begin{align}
\mathcal{B}_t(s_i) = \bigcup_{s_j \in \mathcal{B}_{t-1}(s_i)} S_j,
\end{align}
with $\mathcal{B}_0(s_i)=s_i$. At time $t$, the function $\mathcal{B}_t(s_i)$ represents the neighbours (of the neighbours)$\times (t-1)$ of node $s_i$.
\end{definition}

This section assumes that nodes have the same initial transmit probability $f$ and the same erasure probability $q$. The following lemma links the estimator $\tilde{Z}(t)$ given by \algref{alg:collection} to the BLUE of the network size $\tilde{n}$ for $t=0,1$ and $t=\infty$:

\begin{lemma}
The estimator $\tilde{Z}(t)$ for $t=0,1$ and $t=\infty$ is proportional to the BLUE of the network size $\tilde{n}$. In other words, it can be written as follows:
\begin{align}
\tilde{Z}(0) = N_{\max} \alpha_0 \tilde{n} \nonumber \\
\tilde{Z}(1) = N_{\max} \alpha_0 \alpha_1 \tilde{n} \nonumber \\
\lim_{t \rightarrow \infty} \tilde{Z}(t) = \tilde{n},
\end{align}
where $\alpha_0 = \cfrac{K}{N_{\max}^2}$ and $\alpha_1 =(1+ \cfrac{N_{\max}-K}{LW}\pi R^2f(1-q))$.
\label{l6}
\end{lemma}
\begin{proof}
The proof can be found in \appref{ap6}.
\end{proof}

From the expressions proposed in \lref{l6}, it is clear that when the number of queried nodes $K = N_{\max}$, then the estimator $\tilde{Z}(t)$ is equal to the BLUE of the network size. Such property linking the counting estimator to the BLUE is conjectured to be valid of all time instant $t$:
\begin{conjecture}
The estimator $\tilde{Z}(t)$ is the proportional to the BLUE $\tilde{n}$ of the network size and can be written as:
\begin{align}
\tilde{Z}(t) = N_{\max} \prod_{k=0}^t \alpha_k \tilde{n},
\end{align}
with
\begin{align}
\prod_{k=0}^\infty \alpha_k &= 1/N_{\max} \text{ and } \lim_{K \rightarrow N_{\max}} \prod_{k=0}^t \alpha_k &= 1/N_{\max}.
\end{align}
\end{conjecture}

\section{Simulation Results} \label{sec:sim}

This section presents the simulation results of the proposed counting algorithm. In all the simulations, the bound is set to $N_{\max}=350$ for a network containing $n=300$ nodes. The field is set to the unit square, the connectivity radius to $R=0.1$ and the average packet erasure to $Q=0.1$. Due to space limitations, the performance of the network size estimator is not presented.

\begin{figure}[ht]
\centering
\includegraphics[width=0.92\linewidth]{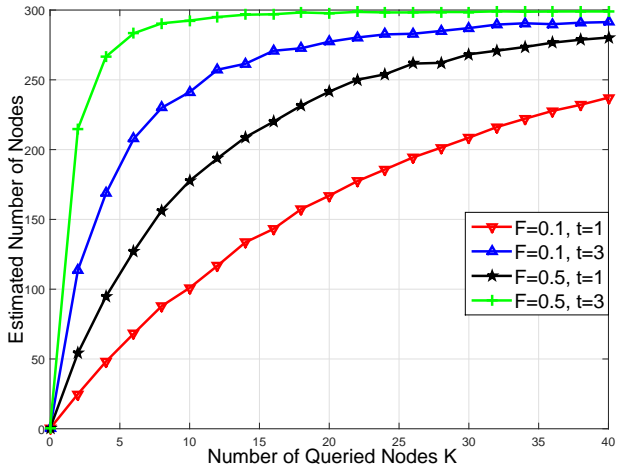}
\caption{Number of Queried nodes versus the number of estimated nodes for different combination of query time $t$ and initial transmit probability $F$. The network contains $n=300$ nodes bounded by $N_{\max}=350$. The connectivity radius is $R=0.1$ and the average erasure $Q=0.1$.}\label{fig:NoQuriednodesTF}
\end{figure}

\begin{figure}[ht]
\centering
\includegraphics[width=0.92\linewidth]{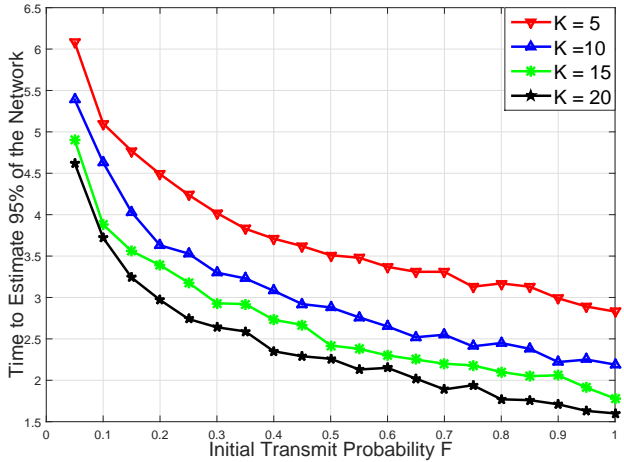}\\
\caption{Initial transmit probability $F$ versus the average time to perform $95\%$ estimation of the network for different number of queried nodes $K$. The network contains $n=300$ nodes bounded by $N_{\max}=350$. The connectivity radius is $R=0.1$ and the average erasure $Q=0.1$.}\label{fig:probFtimeT}
\end{figure}

\begin{figure}[ht]
\centering
\includegraphics[width=0.92\linewidth]{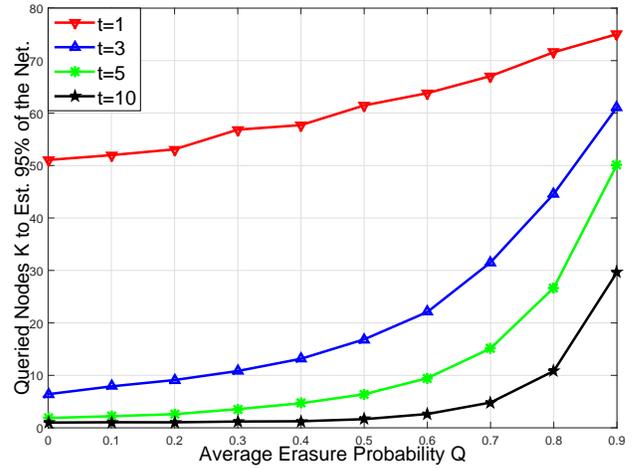}\\
\caption{Number of queried nodes $K$ versus the average erasure probability to perform $95\%$ estimation of the network for different running times $t$. The network contains $n=300$ nodes bounded by $N_{\max}=350$. The connectivity radius is $R=0.1$.}\label{fig:probKtimeF}
\end{figure}

\begin{figure}[ht]
\centering
\includegraphics[width=0.92\linewidth]{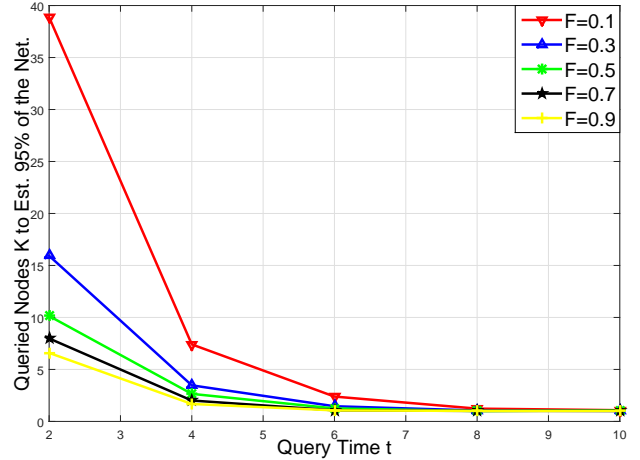}\\
\caption{Query time $t$ versus the average number of queried nodes to perform $95\%$ estimation of the network for different initial transmit probability $F$. The network contains $n=300$ nodes bounded by $N_{\max}=350$. The connectivity radius is $R=0.1$ and the average erasure $Q=0.1$.}\label{fig:timeTQueriedK}
\end{figure}

\fref{fig:NoQuriednodesTF} shows the relation between the number of queried nodes and number of estimated nodes in the network at varied query time $t$ and transmit probability $F$. We notice that asking $10\%$ or more of nodes gives a good estimation of the network size. Besides, increasing the initial transmit probability $F$ or the query time $t$ results in an enhancement of the performances.

\fref{fig:probFtimeT} shows the initial transmit probability $F$ versus the time $t$, in which the total estimation of network nodes is $95\%$ or more for various queried nodes $K$. One can notice that for fixed $F=0.5$, increasing the queried nodes from $K=10$ to $K=20$, reduces the average time $t$ to disseminate the node's information in the network.

\fref{fig:probKtimeF} illustrates the number of queried nodes $K$ versus the average erasure probability, in which the total estimation of network nodes is $95\%$ or more for different running times $t$. As expected, the number of queried nodes to perform $95\%$ estimation of the whole network size decreases with the number of iteration of the algorithm. This can be explained by the fact that as the number of iteration increases, each node have more knowledge about the network configuration that results in a less queried nodes.

\fref{fig:timeTQueriedK} shows the relationship between the queried time $t$ versus the mean number of queried nodes $K$ to achieve $95\%$ or more of the total estimation of network size. We first note that for $t=8$, the counting estimator reached the BLUE. Hence, for our setting, $t=8$ is sufficient for the condition $t \rightarrow \infty$. We also note that increasing the initial transmit probability results in an improvement the estimation of the network size.

\section{Conclusion} \label{sec:con}

This work introduces a novel hybrid size estimation algorithm in an anonymous graph, in which each node knows only its unique identifier. A node counting algorithm is proposed whose output can be used to perform network size estimation using statistical inference methods. For dense graphs and accurate running times, the paper shows that the proposed algorithm produces an estimate of the total number of nodes proportional to the BLUE and that it converges when all the network nodes are queried. Simulation results show that the proposed algorithm produces a good estimate when either the running time or the number of queried nodes are reasonable. As a future research direction, the proposed conjecture can be demonstrated, and the result of the paper can be extended to networks with nodes having non-unique IDs or non-maintaining fixed network topology.

\appendices

\numberwithin{equation}{section}

\section{Proof of \lref{l6}} \label{ap6}

This section provides the proof of \lref{l6}. The proofs rely on auxiliary results of \thref{th1}, \thref{th2}, \lref{l1}, and \lref{l2} that are available in \appref{ap1}.

\subsection{Performance for $\tilde{Z}(0)$}

\begin{lemma}
The estimator $\tilde{Z}(t)$ for $t=0$ is the proportional to the BLUE $\tilde{n}$ of the network size. In other words, we have:
\begin{align}
\tilde{Z}(0) = N_{\max} \alpha_0 \tilde{n},
\end{align}
where $\alpha_0 = K/N_{\max}^2$.
\label{l3}
\end{lemma}

\begin{proof}
At time $t=0$, from the initialisation part of the packets $P_{s_i}$ of a node $s_i \in \mathcal{A}$ in \algref{alg:initial}, we have:
\begin{align}
\mathds{P}(X_{ij}(0) = 1) =
\begin{cases} 
1 \hspace{2cm} &\text{if } j=i \\
0 \hspace{2cm} &\text{otherwise}. 
\end{cases}
\end{align}

Hence, for an arbitrary node $s_i \in \mathcal{N}$, we have:
\begin{align}
\mathds{P}(X_{ii}(0) = 1) &= \mathds{P}(X_{ii}(0) = 1| s_i \in \mathcal{A})\mathds{P}(s_i \in \mathcal{A}) \nonumber \\ 
& + \mathds{P}(X_{ii}(0) = 1| s_i \in \mathcal{D})\mathds{P}(s_i \in \mathcal{D}).
\end{align}
Since $ \mathds{P}(X_{ii}(0) = 1| s_i \in \mathcal{D})=0$, then $\mathds{P}(X_{ii}(0) = 1) = \cfrac{n}{N_{\max}}$. Hence, we can write: 
\begin{align}
Z_i (0) &= \max_{s_j \in \mathcal{K}} X_{ji}(0) \nonumber \\
&=\begin{cases}
X_{ii}(t) \hspace{2cm} &\text{if } s_i \in \mathcal{K} \\
0 \hspace{2cm} &\text{otherwise}. 
\end{cases}
\end{align}
Therefore, we obtain:
\begin{align}
\mathds{P}(Z_{i}(0) = 1) &= \mathds{P}(Z_{i}(0) = 1| s_i \in \mathcal{K})\mathds{P}(s_i \in \mathcal{K}) \nonumber \\ 
& + \mathds{P}(Z_{i}(0) = 1| s_i \notin \mathcal{K})\mathds{P}(s_i \notin \mathcal{K}) \nonumber \\ 
&= \mathds{P}(X_{ii}(0) = 1)\mathds{P}(s_i \in \mathcal{K}) = \cfrac{nK}{N_{\max}^2} \nonumber \\
&= n \alpha_0 .
\end{align}
Using \thref{th2}, the estimator $\tilde{Z}$ is proportional to the BLUE estimate of $n$.
\end{proof}

\subsection{Performance for $\tilde{Z}(1)$}

\begin{lemma}
The estimator $\tilde{Z}(t)$ for $t=1$ is the proportional to the BLUE $\tilde{n}$ of the network size. In other words, we have:
\begin{align}
\tilde{Z}(1) = N_{\max} \alpha_0 \alpha_1 \tilde{n},
\end{align}
where $\alpha_0 = \cfrac{K}{N_{\max}^2}$ and $\alpha_1 =(1+ \cfrac{N_{\max}-K}{LW}\pi R^2f(1-q))$.
\label{l4}
\end{lemma}

\begin{proof}
At time $t=1$, from the initialisation part of the packets $P_{s_i}$ of a node $s_i \in \mathcal{A}$ in \algref{alg:initial} and \lref{l1}, we have:
\begin{align}
\mathds{P}(X_{ij}(0) = 1) =
\begin{cases} 
1 \hspace{0.5cm} &\text{if } j=i \\
p_{ij} \hspace{0.5cm} &\text{if } s_j \in \mathcal{B}_1(s_i) \setminus \mathcal{B}_0(s_i)\\
0 \hspace{0.5cm} &\text{otherwise}. 
\end{cases}
\end{align}

For node a $s_j \in \mathcal{B}_1 \setminus \mathcal{B}_0$, node $s_i \in \mathcal{A}$ knows it is alive if the two following events occur:
\begin{itemize}
\item Node $s_j$ transmit its packet. This event happens with probability $f_{s_j}$.
\item The packet transmitted from $s_j$ to $s_i$ is successfully received. This event happens with probability $1-q_{ji}$.
\end{itemize}

Given that the events are independent, hence the probability $p_{ij}$ can be expressed as $p_{ij}=f_{s_j}(1-q_{ji})$. The probability that a node $s_i$ knows that a node $s_j$ is alive can therefore be expressed as:
\begin{align}
&\mathds{P}(X_{ij}(0) = 1) = \mathds{P}(X_{ij}(0) = 1| s_i \in \mathcal{A})\mathds{P}(s_i \in \mathcal{A}) \nonumber \\ 
& \qquad \qquad + \mathds{P}(X_{ij}(0) = 1| s_i \in \mathcal{D})\mathds{P}(s_i \in \mathcal{D}) \nonumber \\ 
&=\mathds{P}(X_{ij}(0) = 1| s_i \in \mathcal{A})\cfrac{n}{N_{\max}}
\nonumber \\ 
& =\cfrac{n}{N_{\max}} 
\begin{cases}
1 \hspace{0.5cm} &\text{if } j=i \\
f_{s_j}(1-q_{ji}) \hspace{0.5cm} &\text{if } s_j \in \mathcal{B}_1(s_i) \setminus \mathcal{B}_0(s_i)\\
0 \hspace{0.5cm} &\text{otherwise}. 
\end{cases}
\end{align}

We obtain the expression of $Z_i(t)$ for $t=1$ as follows:
\begin{align}
\mathds{P}(Z_{i}(0) = 1) &= \mathds{P}(Z_{i}(0) = 1| s_i \in \mathcal{K})\mathds{P}(s_i \in \mathcal{K}) \nonumber \\ 
& + \mathds{P}(Z_{i}(0) = 1| s_i \notin \mathcal{K})\mathds{P}(s_i \notin \mathcal{K}) \nonumber \\ 
&= \cfrac{nK}{N_{\max}^2} + \cfrac{N_{\max}-K}{N_{\max}}\mathds{P}(Z_{i}(0) = 1| s_i \notin \mathcal{K}).
\end{align}

The second term can be expressed as:
\begin{align}
&\mathds{P}(Z_{i}(0) = 1| s_i \notin \mathcal{K}) = \nonumber \\
&\sum_{s_j \in \mathcal{K}} \mathds{P}(Z_{i}(0) = 1| s_i \in \mathcal{B}_j(1) \setminus \mathcal{B}_j(0))\mathds{P}(s_i \in \mathcal{B}_j(1) \setminus \mathcal{B}_j(0)) \nonumber \\
&+ \mathds{P}(Z_{i}(0) = 1| s_i \notin \bigcup_{s_j \in \mathcal{K}}\mathcal{B}_j(1)) \mathds{P}( s_i \notin \bigcup_{s_j \in \mathcal{K}}\mathcal{B}_j(1)).
\end{align}

Note that we removed the conditioning $ s_i \notin \mathcal{K}$ only for clarity. We first compute $\mathds{P}(s_i \in \mathcal{B}_j(1) \setminus \mathcal{B}_j(0))$. From the connectivity condition of two nodes in the network, the probability can be expressed as $\mathds{P}(d(s_i,s_j)<R)$, where $R$ is the connectivity radius. The nodes are uniformly distributed in a rectangle of width $W$ and length $L$. Therefore, we have:
\begin{align}
\mathds{P}(d(s_i,s_j)<R) = \cfrac{\pi R^2}{LW}
\end{align}
The term can be simplified as
\begin{align}
&\mathds{P}(Z_{i}(0) = 1| s_i \notin \mathcal{K}) =  \cfrac{\pi R^2}{LW} \cfrac{n}{N_{\max}}  \sum_{s_j \in \mathcal{K}} f_{s_i}(1-q_{ij}).
\end{align}
If all the node have the same erasure probability and initial transmit probability, the term can further be simplified as:
\begin{align}
&\mathds{P}(Z_{i}(0) = 1| s_i \notin \mathcal{K}) =  \cfrac{\pi R^2}{LW} \cfrac{Kn}{N_{\max}}  f(1-q).
\end{align}
The probability that node $s_i$ is alive can therefore be written as:
\begin{align}
\mathds{P}(Z_{i}(0) = 1) &= \cfrac{nK}{N_{\max}^2} (1+ \cfrac{N_{\max}-K}{LW}\pi R^2f(1-q)) \nonumber \\
&= n \alpha_0 (1+ \cfrac{N_{\max}-K}{LW}\pi R^2f(1-q)) \nonumber \\
&= n \alpha_0 \alpha_1.
\end{align}
\end{proof}

\subsection{Performance for $\tilde{Z}(\infty)$}

\begin{lemma}
The limit of the BLUE $\tilde{n}$ of the network size goes to $\tilde{Z}(t)$ as $t$ goes to $\infty$. In other words, we have:
\begin{align}
\lim_{t \rightarrow \infty} \tilde{Z}(t) = \tilde{n}.
\end{align}
\label{l5}
\end{lemma}

\begin{proof}
To proof this lemma, we first compute the MLE estimator $\tilde{n}$ of $n$ as $t\rightarrow \infty$. From \lref{l2}, we note that the average number of alive neighbours of an arbitrary alive node is an increasing function. If we assumed that the whole network is connected, then the average number of alive neighbours of an arbitrary alive node is a strictly increasing function bounded by $n$. Therefore, $\exists \ t_0$ such that $\forall \ t \geq t_0$, we have $\mathcal{B}_t(s_i) = \mathcal{A}, \forall \ s_i \in \mathcal{A}$. Given that $\mathds{P}(X_{ij}(t) = 1)>0, \ \forall \ s_j \in \mathcal{A}, t \geq t_0$ and that is a strictly increasing function bounded by $1$, then $\exists \ t^*$ such that for $s_j \in \mathcal{A}$:

\begin{align}
\mathds{P}(X_{ij}(t^*) = 1) =
\begin{cases}
1 \hspace{2cm} &\text{if } j=i \\
\cfrac{n}{N_{\max}}  \hspace{2cm} &\text{otherwise}. 
\end{cases}
\end{align}
We can write:
\begin{align}
\mathds{P}(Z_i(t^*) = 1) &= \mathds{P}(Z_i(t^*) = 1 | s_i \in \mathcal{A})\mathds{P}(s_i \in \mathcal{A}) \nonumber \\
& +  \mathds{P}(Z_i(t^*) = 1 | s_i \notin \mathcal{A})\mathds{P}(s_i \notin \mathcal{A}).
\end{align}
Since $\mathds{P}(Z_i(t^*) = 1 | s_i \notin \mathcal{A})=0$, then 
\begin{align}
\mathds{P}(Z_i(t^*) = 1) = \cfrac{n}{N_{\max}}\mathds{P}(Z_i(t^*) = 1 | s_i \in \mathcal{A}).
\end{align}
Using \thref{th1}, the probability $\mathds{P}(Z_i(t^*) = 1 | s_i \in \mathcal{A})$ can be written as:
\begin{align}
\mathds{P}(Z_i(t^*) = 1 | s_i \in \mathcal{A}) = 1-\prod\limits_{s_j \in \mathcal{K}}(1-p_{ji}(t^*)),
\end{align}
where $p_{ji}(t^*) = \mathds{P}(X_{ji}(t^*) = 1 | s_i \in \mathcal{A})$. Two cases can be distinguished:
\begin{itemize}
\item $s_i \in \mathcal{K}$: By substituting $p_{ii} = 1$, we have:$\mathds{P}(Z_i(t^*) = 1 | s_i \in \mathcal{A}) = 1$.
\item $s_i \notin \mathcal{K}$: By substituting $p_{ji} = \cfrac{n}{N_{\max}}$, we have:$\mathds{P}(Z_i(t^*) = 1 | s_i \in \mathcal{A}) = 1 - \left(\cfrac{N_{\max}-n}{N_{\max}}\right)^{K}$. For dense networks, we have $\left(\cfrac{N_{\max}-n}{N_{\max}}\right)^{K} \approx 0$, hence $\mathds{P}(Z_i(t^*) = 1 | s_i \in \mathcal{A}) = 1$.
\end{itemize}
In both cases, we obtain $\mathds{P}(Z_i(t^*) = 1) = \cfrac{n}{N_{\max}} = \alpha n, \ \forall \ t \geq t^*$. Another alternative is to assume that among the $K$ queried nodes, at least one of the node is alive. In that case, we directly obtain $\mathds{P}(Z_i(t^*) = 1) = \alpha n, \ \forall \ t \geq t^*$. Using \thref{th2}, the BLUE estimator $\tilde{n}$ of $n$ can be written as:
\begin{align}
\tilde{n} = \cfrac{\sum_{i=1}^{N_{\max}}Z_i}{N_{\max} \alpha} = \sum_{i=1}^{N_{\max}}Z_i = \tilde{Z}.
\end{align}
\end{proof}

\section{Auxiliary Results}\label{ap1}

\subsection{Maximum of Bernoulli Random Variables}

\begin{theorem}
Let $X_1, \ \cdots, X_{n}$ be independent Bernoulli random variable with $\mathds{P}(X_i=1)=p_i,\ 1 \leq i \leq n$. The random variable $Z=\max_{1 \leq i \leq n} X_i$ is a Bernoulli random variable with parameter $p=\mathds{P}(Z=1)=1-\prod\limits_{i=1}^n(1-p_i)$.
\label{th1}
\end{theorem}

\begin{proof}
Since the only possible values of $X_i,\ 1 \leq i \leq n$ are $0$ and $1$, then the support of $Z$ is $\{0,1\}$. We can clearly see that:
\begin{align}
\mathds{P}(Z=0) &= \mathds{P}(X_1=0, \ \cdots, X_{n}=0) \\ \nonumber
 &\overset{\text{ind}}= \prod\limits_{i=1}^n \mathds{P}(X_i=0) = \prod\limits_{i=1}^n(1-p_i).
\end{align}
Therefore, the random variable $Z=\max_{1 \leq i \leq n} X_i$ is a Bernoulli random variable with parameter $p=1-\prod\limits_{i=1}^n(1-p_i)$.
\end{proof}

\subsection{Best Linear Unbiased Estimator of Bernoulli Random Variables}

\begin{theorem}
Let $X_1, \ \cdots, X_{N}$ be identical independent Bernoulli random variable with $\mathds{P}(X_i=1)=n \alpha,\ 1 \leq i \leq N$ where $\alpha$ is a constant that do not depend on $n$. The quantity $\sum\limits_{i=1}^{N}X_i$ is proportional to the Maximum Likelihood Estimator (MLE) of the quantity $n$. Moreover, the estimator $\tilde{n}= \cfrac{\sum_{i=1}^{N}X_i}{N \alpha}$ is the best linear unbiased estimator (BLUE) of the quantity $n$.
\label{th2}
\end{theorem}

\begin{proof}
The likelihood function of $(X_1, ,\ \cdots, \ X_{N})$ can be written as:
\begin{align}
&f_{X_1, ,\ \cdots, \ X_{N}}(x_1, ,\ \cdots, \ x_{N}) = \prod_{i=1}^{N} (n \alpha)^{x_i} (1-n \alpha)^{1-x_i} \nonumber \\
&\text{log}(f_{X_1, ,\ \cdots, \ X_{N}}) \propto \sum_{i=1}^{N}x_i \text{log}(n) + (1-x_i) \text{log}(1-n \alpha).
\end{align}
Solving the equation $\cfrac{d}{dn}\text{log}(f_{X_1, ,\ \cdots, \ X_{N}})=0$ yields the following MLE:
\begin{align}
\tilde{n} = \cfrac{\sum_{i=1}^{N}X_i}{N \alpha} \ .
\end{align}
Therefore, $\sum\limits_{i=1}^{N}X_i$ is proportional to the MLE of the quantity $n$. The mean of $\tilde{n}$ can be expressed as:
\begin{align}
\mathds{E}(\tilde{n}) = \cfrac{\sum_{i=1}^{N}\mathds{E}(X_i)}{N \alpha} = \cfrac{\sum_{i=1}^{N}n \alpha}{N \alpha} = n.
\end{align}
which conclude that the estimator is unbiased. The variance can be obtained as follows:
\begin{align}
\text{Var}(\tilde{n}) &= \cfrac{\sum_{i=1}^N \text{Var}(X_i)}{N^2 \alpha^2} = \cfrac{\sum_{i=1}^{N}(n \alpha)(1-n \alpha)}{N^2 \alpha^2} \nonumber \\
& = \cfrac{n(1-n \alpha)}{N \alpha} \ .
\end{align}
Computing the Fisher information yields:
\begin{align}
\mathds{E}\left(-\cfrac{d^2}{dn^2}\text{log}(f_{X_1, ,\ \cdots, \ X_{N}})\right) &= \mathds{E}\left(\cfrac{\sum_{i=1}^{N}X_i}{n^2} \right) \nonumber \\
& \quad+ \mathds{E}\left(\cfrac{\alpha n \sum_{i=1}^{N}(1-X_i)}{(1-n\alpha)^2}\right) \nonumber \\
&= \cfrac{N \alpha}{n} + \cfrac{N n\alpha}{1-n \alpha} \nonumber \\
&= \cfrac{N \alpha}{n(1-n \alpha)} \ .
\end{align}
Finally, $\tilde{n}$ is the BLUE of the quantity $n$.
\end{proof}

\subsection{Maximum Number of Reachable Nodes}

\begin{lemma}
The maximum number of neighbours node $s_i$ can know at time instant $t$ is $|\mathcal{B}_t(s_i)|$.
\label{l1}
\end{lemma}

\begin{proof}
The proof is a direct consequence of the data dissemination \algref{alg:knowledge}. At each round, each node can transmit to all of its neighbours. Hence after $t$ round, the information initialled at a node $s_i$ would have travelled at most inside $\mathcal{B}_t(s_i)$. Due to the symmetry of the problem, the farther information node $s_i$ can get is initialled inside $\mathcal{B}_t(s_i)$ which conclude that the maximum number of neighbours node $s_i$ can know at time instant $t$ is $|\mathcal{B}_t(s_i)|$.
\end{proof}

\subsection{Average Number of $t$-degree Neighbours}

\begin{lemma}
The average number of nodes in $\mathcal{B}_t(s_i) \setminus \mathcal{B}_{t-1}(s_i), t \geq 1$ can be approximated by:
\begin{align}
|\mathcal{B}_t(s_i) \setminus \mathcal{B}_{t-1}(s_i)| = \cfrac{n \pi R^2(2t-1)}{LW} .
\end{align}
\label{l2}
\end{lemma}

\begin{proof}
To proof this lemma, we first show that $|\mathcal{B}_t(s_i)| = \cfrac{n\pi R^2t}{LW}$. Using the fact that $\mathcal{B}_k(s_i) \subseteq \mathcal{B}_t(s_i),\ \forall \ k \leq t$, we can write $|\mathcal{B}_t(s_i) \setminus \mathcal{B}_{t-1}(s_i)| = |\mathcal{B}_t(s_i) |-| \mathcal{B}_{t-1}(s_i)|$ which conclude the proof. 

We proof that $|\mathcal{B}_t(s_i)| = \cfrac{n \pi (tR)^2}{LW}$ by induction. For $t=1$, we can clearly see from the definition of $\mathcal{B}_1(s_i)$ that $s_j \in \mathcal{B}_1(s_i)$ if and only if $d(s_i,s_j) \leq R$. Since the nodes are uniformly distributed over $[0,L][0,W]$ and neglecting the side effects, the average number of nodes is $\cfrac{n\pi R^2}{LW}$. Assume the preposition hold for $t$ and that $s_j \in \mathcal{B}_t(s_i)$ if and only if $d(s_i,s_j) \leq tR$. Assume $\exists s_j \in \mathcal{B}_{t+1}(s_i)$ such that $d(s_j,s_i)>(t+1)R$. From the triangular inequalities of the distance operator, we can write for all node $s_k \in \mathcal{B}_t(s_i)$:
\begin{align}
d(s_i,s_j) \leq d(s_i,s_k) + d(s_k,s_j)  .
\end{align}
From the assumption at step $t$, we have $d(s_i,s_j) \leq tR$. Therefore, we obtain:
\begin{align}
d(s_k,s_j) \geq d(s_i,s_j) - d(s_i,s_k) > (t+1)R - tR = R.
\end{align}

Since for nodes $s$ and $s^{\prime}$ to be connected, we should have $d(s,s^\prime) \leq R$, then node $s_j$ is not connected to any node $s_k \in \mathcal{B}_t(s_i)$. Therefore, $d(s_j,s_i) \leq (t+1)R$. This last expression translates to the fact that the average number of nodes in $\mathcal{B}_{t+1}(s_i)$ is $\cfrac{n\pi (t+1)^2R^2}{LW}$. Finally, using the fact that $|\mathcal{B}_t(s_i) \setminus \mathcal{B}_{t-1}(s_i)| = |\mathcal{B}_t(s_i) |-| \mathcal{B}_{t-1}(s_i)|$, we conclude that 
\begin{align}
|\mathcal{B}_t(s_i) \setminus \mathcal{B}_{t-1}(s_i)| = \cfrac{n\pi R^2(2t-1)}{LW}.
\end{align}
\end{proof}

\bibliographystyle{IEEEtran}
\bibliography{citations}

\end{document}